\documentclass[envcountsame]{llncs}
\pdfoutput=1
\usepackage{macros}

\tikzset{
    >=stealth',
    fa/.style={
           rectangle,
           rounded corners,
           draw=black, 
           text width=6em,
           minimum height=2em,
           text centered},
    st/.style={
					 circle,
           draw=black, 
           text centered
           }
}

\usepackage{centernot}

\newcommand{\abs}[1]{\left| #1 \right|}

\newcommand{\BigO}{\mathcal{O}}

\DeclareMathOperator*{\sw}{\preccurlyeq}
\DeclareMathOperator*{\subw}{\sw}
\DeclareMathOperator*{\supw}{\succcurlyeq}
\DeclareMathOperator*{\symdiff}{\oplus}

\newcommand{\dc}[1]{\nabla #1}
\newcommand{\uc}[1]{\Delta #1}

\DeclareMathOperator*{\dep}{\triangleright}
\DeclareMathOperator*{\depeq}{\trianglerighteq}
\DeclareMathOperator*{\deppo}{\trianglerighteq^\ast}
\DeclareMathOperator*{\ndeppo}{{\centernot\trianglerighteq}^\ast}

\newcommand{\lang}{\mathcal{L}}

\newcommand{\PSPACE}{\textsf{PSPACE}}
\newcommand{\NP}{\textsf{NP}}
\newcommand{\cNP}{\textsf{coNP}}
\newcommand{\TAUT}{\textsf{TAUT}}

\newcommand{\cfgana}{cfg-analyzer\xspace}
\newcommand{\fpsolve}{\textsc{FPsolve}\xspace}

\newcommand{\dceqp}{\stackrel{?}{\equiv}_{\dc{}}}
\newcommand{\uceqp}{\stackrel{?}{\equiv}_{\uc{}}}

\newcommand{\bl}[2]{ {#1}^{(#2)} }

\begin{document}

\title{Finite Automata for the Sub- and Superword Closure of CFLs: Descriptional and Computational Complexity\thanks{This work was partially funded by the DFG project ``Polynomial Systems on Semirings: Foundations, Algorithms, Applications''.}}

\author{Georg Bachmeier \and Michael Luttenberger \and  Maximilian Schlund}
\institute{Technische Universit\"{a}t M\"{u}nchen, \email{\{bachmeie,luttenbe,schlund\}@in.tum.de}
}

\maketitle

\begin{abstract}
We answer two open questions by (Gruber, Holzer, Kutrib, 2009) on the state-complexity of representing 
sub- or superword closures of context-free grammars (CFGs):
(1) We prove a (tight) upper bound of $2^{\BigO(n)}$ on the size of nondeterministic finite automata (NFAs) representing the subword closure of a CFG of size $n$.
(2) We present a family of CFGs for which the minimal deterministic finite automata representing their subword closure matches the upper-bound of $2^{2^{\BigO(n)}}$ following from (1).
Furthermore, we prove that the inequivalence problem for NFAs representing sub- or superword-closed languages
is only NP-complete as opposed to PSPACE-complete for general NFAs.
Finally, we extend our results into an approximation method to attack inequivalence problems
for CFGs. %and report experimental results.
\end{abstract}

\section{Introduction}

Given a (finite) word $w= w_1w_2\ldots w_n$ over some alphabet $\Sigma$, we say that $u$ is a {\em (scattered) subword or subsequence} of $w$ if $u$ can be obtained from $w$ by 
erasing some letters of $w$.
We denote the fact that $u$ is a subword of $w$ by $u\sw w$, and alternatively say that $w$ is a {\em superword} of $u$.
As shown by Higman \cite{Higman52} in 1952 $\sw$ is a well-quasi-order on $\Sigma^*$, implying that
{\em every} language $L\subseteq \Sigma^\ast$ has a finite set of $\sw$-minimal elements.
This proves that 
both the subword (also: downward) closure 
$\dc{L} := \{u\in \Sigma^\ast \mid \exists w\in L\colon u \sw w\}$ and the superword (also: upward) closure 
$\uc{L} := \{ w\in\Sigma^\ast \mid \exists u\in L\mid u \sw w\}$ are regular for \emph{any} language $L$.
While in general, 
we cannot effectively construct a finite automaton accepting $\dc{L}$ resp.\ $\uc{L}$,
for specific classes of languages effective constructions are known.

It is well-known that this is the case when when $L$ is given as a context-free grammar (CFG).
This was first shown by van Leeuwen \cite{vLeeuwen78} in 1978.
Later, Courcelle gave an alternative proof of this result in \cite{Courcelle91}.
Section \ref{sec:upperbound} builds up on these results by Courcelle.
We also mention that for Petri-net languages an effective construction is known thanks to Habermehl, Meyer, and Wimmel \cite{DBLP:conf/icalp/HabermehlMW10}.

These results can be used to tackle undecidable questions regarding the ambiguity, inclusion, equivalence, universality or emptiness of languages by 
over-approximating one or both languages by suitable regular languages \cite{MohriNederhof01,DBLP:conf/fase/LongCMM12,DBLP:journals/fmsd/GantyMM12,DBLP:conf/icalp/HabermehlMW10}:
For instance, consider the scenario where we are given a procedural program whose 
runs can be described as a pushdown automaton resp.\ a CFG $G_1$
and a context-free specification $G_2$ of all safe executions, and we want to check whether all runs of the system conform to the 
safety specification $\lang(G_1) \subseteq \lang(G_2)$. As $\lang(G_1)\cap \overline{\dc{\lang(G_2)}} \neq \emptyset \Rightarrow \lang(G_1) \not\subseteq \lang(G_2)$,
we can obtain at least a partial answer to the otherwise undecidable question.
Of course, in the case $\lang(G_1)\subseteq \dc{\lang(G_2)}$
no information is gained, and one needs to refine the problem e.g.\ 
by using some sort of counter-example guided abstraction refinement as done e.g.\ in \cite{DBLP:conf/fase/LongCMM12}.

\paragraph*{Contributions and Outline}
Our first results (Sections \ref{sec:upperbound} and \ref{sec:debu})
concern the blow-up incurred when constructing a (non-)deterministic finite automaton (NFA resp.\ DFA) for the subword closure of a language given by a context-free grammar $G$ where we improve the results of \cite{Gruber:2009:MSH:1551570.1551577}:
For a CFG $G$ of size $n$, \cite{Gruber:2009:MSH:1551570.1551577} shows that 
an NFA recognizing $\dc\lang(G)$ has at most $2^{2^{\BigO(n)}}$ states, and there are CFGs requiring
at least $2^{\Omega(n)}$ states. (For linear CFGs the upper and lower bounds are both single exponential.)
The upper bound of \cite{Gruber:2009:MSH:1551570.1551577} is established by analyzing the inductive construction of \cite{vLeeuwen78}. We improve this result in Section \ref{sec:upperbound} to $2^{\BigO(n)}$ by slightly adapting Courcelle's construction \cite{Courcelle91} (we also briefly discuss that naively applying Courcelle's construction cannot do better than $2^{\Omega(n\log n)}$ in general). This result of course yields immediately an upper bound of $2^{2^{\BigO(n)}}$ on the size of minimal DFA representing accepting $\dc\lang(G)$. In Section~\ref{sec:debu} we show this bound is tight already over a binary alphabet. To the best of our knowledge, so far only examples were known which showcase the single-exponential blow-up when constructing an NFA accepting the subword closure of a context-free grammar\cite{Gruber:2009:MSH:1551570.1551577} resp.\ 
a DFA accepting the subword closure of a DFA or NFA \cite{DBLP:journals/fuin/Okhotin10}. 
We then study in Section \ref{sec:equiv} the equivalence problem for NFAs recognizing subword- resp.\ supword-closed languages.
While for general NFAs this problem is \PSPACE-complete, we show that it becomes \cNP-complete under this restriction.
We combine these results in Section \ref{sec:application} to derive a conceptual simple semi-decision procedure for checking language-inequivalence of two CFGs $G_1,G_2$: we first construct NFAs for $\dc{\lang(G_1)}$ and $\dc\lang(G_2)$, and check language-inequivalence of these NFAs; if the NFAs are inequivalent, we construct a witness of the language-inequivalence of $G_1$ and $G_2$; otherwise we refine the grammars, and repeat the test on the so obtained new grammars. This approach is motivated by the abstraction-refinement approach of \cite{DBLP:conf/fase/LongCMM12} for checking if the intersection of two context-free languages is empty.
We experimentally evaluate our approach by comparing it to {\em \cfgana} of \cite{DBLP:conf/icalp/AxelssonHL08} which uses incremental SAT-solving
 %and an efficient encoding of the CYK-algorithm in propositional logic 
to tackle the language-inequivalence problem.

\section{Preliminaries}
By $\Sigma$ we denote a finite alphabet.
For every natural number $n$, let $\Sigma^{\le n}$ denote the words of length at most $n$ over $\Sigma$.
The empty word is denoted by $\ew$; the set of all finite words by $\Sigma^\ast$.

We measure the \emph{size} $\abs{G}$ of a CFG $G$ as the total number of symbols on the right hand sides of all productions.
The size of an NFA is simply measured as the number of states
(this is an adequate measure for a constant alphabet, since the number of transitions is at most quadratic in the number of states).

Throughout the paper we will always assume that all CFGs are reduced, i.e.~do not contain any unproductive or unreachable nonterminals
(any CFG can be reduced in polynomial time).
Let $X$ be a nonterminal in a CFG $G$. We define $\lang(X)$ as the set of all words $w\in\Sigma^\ast$ derivable from $X$.
If $S$ is the start symbol of $G$, then $\lang(G) := \lang(S)$. Moreover, $\Sigma_X \subseteq \Sigma$ denotes the set of all terminals reachable from $X$.
Overloading notation we sometimes write $\dc{X}$ for $\dc{\lang(X)}$.

The dependency graph of a CFG $G$ is the finite graph with nodes the nonterminals of $G$ 
where there is an edge from $X$ to $Y$
if there is a production $X\to \alpha Y \beta$ in $G$.
We say that $X$ {\em depends directly on} $Y$ (written as $X\dep Y$) if $X\neq Y$ and there is an edge from $X$ to $Y$.
The reflexive and transitive closure of $\dep$ is denoted by $\deppo$.
We write $X\equiv Y$ if $X\deppo Y\wedge Y\deppo X$, i.e.\ if $X$ and $Y$ are located in a common strongly-connected component of the dependency graph.
We say that $G$ is strongly connected if the dependency graph is strongly connected.

From \cite{Courcelle91} we recall some useful facts concerning the subword closure:
\begin{lemma}
\label{lem:facts-courcelle}
For any nonterminals $X,Y,Z$ in a CFG $G$ it holds that:
\begin{enumerate}
\item $\dc(\lang(X) \cup \lang(Y)) = \dc{\lang(X)} \cup \dc{\lang(Y)}$
\item $\dc(\lang(X) \cdot \lang(Y)) =\dc{\lang(X)}\cdot \dc{\lang(Y)}$
\item $X \equiv Y \Rightarrow \dc{X} = \dc{Y}$ \label{equiv-scc}
\item If $X \rightarrow^* \alpha Y \beta Z \gamma $ for $Y,Z \equiv X$ then $\dc{X} = \Sigma_X^*$
\end{enumerate}
\end{lemma}

\newcommand{\qen}{q_{\text{en}}}
\newcommand{\qex}{q_{\text{ex}}}

\section{Computing the Subword Closure of CFGs}
\label{sec:upperbound}

In this section we describe an optimized version of the construction in \cite{Courcelle91}
to compute an NFA for the subword closure of a CFG $G$ of size $2^{\BigO(\abs{G})}$,
which is asymptotically optimal. We first illustrate the construction by a simple example.

As explained at the end of the next section, a naive implementation of the construction of \cite{Courcelle91} leads to an automaton of size
$2^{\Omega (n)} n! = 2^{\Omega(n \log n)}$ whereas our approach achieves the (optimal) bound of $2^{\BigO(n)}$.

\subsection{Construction by Example}\label{sec:ex-bound}
Consider the grammar $G$ with start symbol $S$ defined by the productions:

\begin{tabular}{b{7cm}b{3cm}}
$\begin{array}{ll@{\hspace{1cm}}ll}
S &\to XaU \mid UaU \mid X & X &\to ZbY \mid \ew \\[1mm]
Y &\to XYa \mid b & U &\to VZ \mid acb \\[1mm]
V &\to ZU \mid \ew & Z &\to cZ \mid bc \\[1mm]
\end{array}$

&

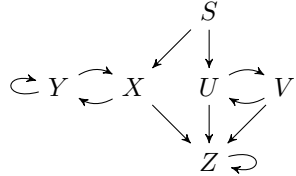
\begin{tikzpicture}
\node (S) at (0,0) {$S$};
\node (X) at (-1,-1) {$X$};
\node (Y) at (-2,-1) {$Y$};
\node (Z) at (0,-2) {$Z$};
\node (U) at (0,-1) {$U$};
\node (V) at (1,-1) {$V$};

\path[->] (S) edge (X)
          edge (U)
      (X) edge[bend left] (Y)
          edge (Z)
      (Y) edge[bend left] (X)
          edge[loop left] (Y)
			(Z) edge[loop right] (Z)
			(U) edge[bend left] (V)
			    edge (Z)
			(V) edge[bend left] (U)
			    edge (Z)
					;
\end{tikzpicture}
\end{tabular}

\noindent
On the right-hand side, the dependency graph is shown where an edge $x\to y$ stands for $x\depeq y$.
To simplify the construction, we first transform the grammar $G$ into a certain normal form $G'$ (with $\dc{\lang(G)}=\dc{\lang(G')}$) and then construct an NFA from $G'$.

In the first step we compute the SCCs of $G$, here $\{X,Y\}$ and $\{U,V\}$.
Since $Y \to XYa$ (with $Y\equiv X$ and $X\equiv X$), we know that $\dc Y = \dc X = \Sigma_X^*= \{a,b,c\}^*$.
We therefore can replace any occurrence of $Y$ by $X$ (thereby removing $Y$ from the grammar) and redefine the rules for $X$ to $X \to aX \mid bX \mid cX \mid \varepsilon$.
In case of the SCC $\{U,V\}$ the grammar is linear w.r.t.\ $U$ and $V$, 
i.e.\ starting from either of the two we can never produce sentential forms in which the total number of occurrences of $U$ and $V$ exceeds one.
Hence, we can identify $U$ and $V$ without changing the subword closure.
Finally, we introduce unique nonterminals for each terminal symbol and restrict the right-hand side of each production to at most two symbols by introducing auxiliary nonterminals $W$ and $T$:
\begin{tabular}{b{9cm}b{4cm}}
$\begin{array}{ll@{\hspace{0.3cm}}ll}
S &\to XW \mid UW \mid X & W &\to A_aU \\[1mm]
X &\to A_aX \mid A_bX \mid A_cX \mid A_\varepsilon & U &\to UZ \mid ZU \mid A_aT \mid A_\varepsilon \\[1mm]
T &\to A_c A_b & Z &\to A_cZ \mid A_bA_c \\[1mm]
A_a &\to a & A_b &\to b\\[1mm]
A_c &\to c & A_\varepsilon &\to \varepsilon\\[1mm]
\end{array}$

&
\scalebox{0.68}
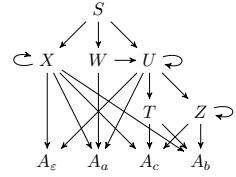
{
\begin{tikzpicture}
\node (S) at (-1,0) {$S$};
\node (X) at (-2,-1) {$X$};
\node (W) at (-1,-1) {$W$};
\node (Z) at (1,-2) {$Z$};
\node (U) at (0,-1) {$U$};
\node (T) at (0,-2) {$T$};
\node (1) at (-2,-3) {$A_{\ew}$};
\node (a) at (-1,-3) {$A_a$};
\node (b) at (1,-3) {$A_b$};
\node (c) at (0,-3) {$A_c$};

\path[->] 
      (S) edge (X)
          edge (W)
					edge (U)
      (X) edge[loop left] (X)
          edge (a)
					edge (1)
					edge (c)
					edge (b)
			(W) edge (a)
			    edge (U)
			(Z) edge[loop right] (Z)
			    edge (c)
					edge (b)
			(U) edge[loop right] (U)
			    edge (Z)
					edge (T)
					edge (a)
					edge (1)
			(T) edge (c)
			    edge (b)
					;
\end{tikzpicture}
}
\end{tabular}

\noindent
Note that the dependency graph of this transformed grammar is now acyclic apart from self-loops.
Because of this, we can directly transform the grammar into an {\em acyclic} equation system (or straight-line program, or algebraic circuit) 
whose solution is a regular expression for $\dc S$:
\[
\begin{array}{ll@{\hspace{1cm}}ll}
\dc A_a &= (a + \varepsilon) & \dc A_b &= (b + \varepsilon)\\
\dc A_c &= (c + \varepsilon) & \dc A_\varepsilon &= \varepsilon\\
\dc Z &= c^*(\dc A_b \dc A_c) & \dc T &= \dc A_c \dc A_b \\
\dc U &= \Sigma_Z^* (\dc A_a \dc T) \Sigma_Z^* & \dc W &= \dc A_a \dc U \\
\dc X &= \Sigma_X^* & \dc S &= \dc X \dc W + \dc U \dc W + \dc X \\
\end{array}
\]
In order to obtain an NFA for $\dc S$, we evaluate this equation system from bottom to top while 
re-using as many of the already constructed automata as possible.
For instance, consider the equation:
\[
\dc S = \dc X \dc W + \dc U \dc W + \ew \cdot \dc X
\]
Because of acyclicity of the equation system, we may assume inductively 
that we have already constructed NFAs $\sA_{\dc X}$, $\sA_{\dc W}$, and $\sA_{\dc U}$ for $\dc X$, $\dc W$, and $\dc U$, respectively.
To construct the NFA for $\dc S$, we first make two copies $\bl{\sA}{1}$, $\bl{\sA}{2}$ of each of these automata.
Automata with superscript $(1)$  will be used exclusively for variable occurrences to the left of the concatenation operator,
while automata with superscript $(2)$ will be used for the remaining occurrences.
We then read quadratic monomials, like $ \dc X \dc W$, as an $\ew$-transition connecting $\bl{\sA}{1}_{\dc X}$ with $\bl{\sA}{2}_{\dc W}$
as shown in Figure~\ref{fig:1} where all edges represent $\ew$-transitions.

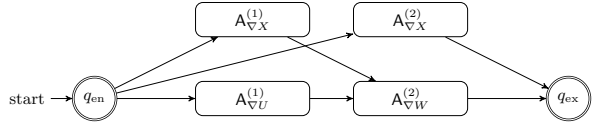
\begin{figure}[h]
\begin{center}
\scalebox{0.7}{
\begin{tikzpicture}
\node[initial,state,accepting] (q) at (0,0) {$\qen$};

\node[fa] (x1) at (3,1.5) {$\bl{\sA}{1}_{\dc X}$};
\node[fa] (u1) at (3,0) {$\bl{\sA}{1}_{\dc U}$};

\node[fa] (x2) at (6,1.5) {$\bl{\sA}{2}_{\dc X}$};
\node[fa] (w2) at (6,0) {$\bl{\sA}{2}_{\dc W}$};

\node[state,accepting] (r) at (9,0) {$\qex$};

\path[->] (q) edge (x1)
              edge (u1)
							edge (x2)
					(x1) edge (w2)
					(u1) edge (w2)
					(x2) edge (r)
					(w2) edge (r)
					;

\end{tikzpicture}
}
\end{center}
\caption{Efficient re-use of re-occuring NFAs in Courcelle's construction.}
\label{fig:1}
\end{figure}

We do not claim that this construction yields the smallest NFA,
but it is easy to describe and yields an NFA of sufficiently small size in order to deduce in the following subsections an
asymptotically tight upper bound on the number of states.
We recall that using a CFG of size $3n+2$ to succinctly represent the singleton language $\{a^{2^n}\}$,
the bound of $2^{\Theta(n)}$ follows \cite{Gruber:2009:MSH:1551570.1551577}.

In \cite{DBLP:conf/concur/AtigBT08} it is remarked that a straight-forward implementation of Courcelle's construction yields an
NFA ``single exponential'' size w.r.t.~$\abs{G}$. However, no detailed complexity analysis is given.
Consider the CFG with start-symbol $A_n$ and consisting of the rules $A_0 \to a$ and for all
$1\leq k \leq n:\quad A_k \to A_iA_j \quad \forall 0\leq i,j \leq (k-1)$.
If we compute an NFA for $\dc{A_n}$ via the straight-forward bottom-up construction it will have size
$a_n := |\sA_{\dc A_n}|$ with
$ 
a_n = 2 + \sum_{0\leq i,j \leq (n-1)} (a_i + a_j).% = 2 + 2n \sum_{i=0}^{n-1} a_i.
$
It is easy to show that
$a_n \geq 2^{n}n! \in 2^{\Omega(n\log n)}$.
Hence, the crucial part to achieve the optimal bound of $2^{\BigO(n)}$ is to reuse already computed automata. % via a ``bipartite wiring''.
We just remark that one can also achieve similar savings, by factoring out common terms in the right hand side
of the acyclic equations.
A subsequent bottom-up construction leads to an NFA of size $2^{\BigO(n)}$ as well but the constant hidden in the $\BigO$ is larger
and the analysis is more involved.
Note that this also shows that we can construct a regular expression of size $2^{\BigO(n)}$
representing the subword closure.

\subsection{Normal Form for Computing the Subword Closure}
To simplify our construction, we will assume that our grammar has a special form which is similar to CNF but with unary rules allowed.
Any CFG can be transformed into this form with at most linear blowup in size preserving its subword closure (but not its language).

\begin{definition}
A CFG $G$ is in quadratic normal form (QNF) if for every terminal $x\in\Sigma\cup \{\varepsilon\}$ there is a unique nonterminal $A_x$ with
the only production $A_x \to x$ and every other production is in one of the following forms:
\begin{itemize}
\item $X\to YX$ or $X \to XY$ (with $Y \neq X$)
\item $X \to Y$ or $X\to YZ$ (with $Y,Z \neq X$)
\end{itemize}
A grammar in QNF is called \emph{simple} if 
\begin{itemize}
\item for all $X \to YX$ or $X \to XY$, we have $X \dep Y$
\item for all $X \to Y$ or $X \to YZ$, we have $X \dep Y,Z$.
\end{itemize}
\end{definition}
Note that the dependency graph associated with a grammar in simple QNF is acyclic with the exception of self-loops.

First, we need a small lemma that allows us to eliminate all linear productions ``within'' some SCC, i.e.~productions of the form
$X \to \alpha Y \beta$ such that $X\neq Y$ but $Y\deppo X$.
\begin{lemma}
\label{lem:scc}
Let $G$ be a strongly connected linear CFG with nonterminals $\vars = \{X_1,\dots,X_n\}$
so that every production is either of the form $X\to \alpha Y \beta$ or $X \to \alpha$ for $\alpha,\beta\in\Sigma^\ast$.

Consider the grammar $G'$ which we obtain from $G$ by replacing in every production of $G$
every occurrence of a nonterminal $X_i$ by $Z$.

We then have that $\dc \lang(Z) = \dc \lang(X_i)$ for all $i\in [n]$. 
\end{lemma}
Using the preceding lemma, we can show that it suffices to consider only CFG in simple QNF in the following.
\begin{theorem}
\label{thm:simpleQNF}
Every CFG $G$ can be transformed into a CFG $G'$ in simple QNF such that $\dc{\lang(G)} = \dc{\lang(G')}$ and $|G'| \in \BigO(|G|)$.
\end{theorem}
\begin{proof}[sketch]
First, we use Lemma \ref{lem:facts-courcelle} to simplify all productions involving an $X$ with $X \Rightarrow^* \alpha X \beta X \gamma$.
Then we apply Lemma \ref{lem:scc} to contract SCCs to a single non-terminal.
Finally, we introduce auxiliary variables for the terminals and we binarize the grammar (keeping unary rules like \cite{DBLP:journals/didactica/LangeL09}).
\end{proof}

\begin{theorem}
\label{thm:upperbound}
For any CFG $G$ in simple QNF with $n$ nonterminals there is an NFA $\sA$ with at most
$2\cdot 3^{n-1}$ states which recognizes the subword closure of $G$, i.e.~$\dc \lang(G) = \lang(\sA)$.
\end{theorem}
\begin{proof}[sketch]
Since the dependency graph of a grammar in simple QNF is a DAG (if we ignore self-loops), we can order the nonterminals according
to a topological ordering of this graph.
We proceed bottom-up to inductively build an NFA for $\dc \lang(G) = \dc S$ as in section \ref{sec:ex-bound}.
Since our grammar is in QNF, at each stage we only have to produce at most two copies of every automaton representing the
subword-closure of a ``lower'' nonterminals $Y$.
Inductively, for each of these $Y$ we can build NFAs with at most $2\cdot 3^i$ many states where $i$ is $Y$'s position in the topological ordering.
Using the ``biparitite wiring'' sketched in Figure~\ref{fig:1} the size of the automaton for $X$ can then be estimated as
\[
\abs{\sA_S} \le 2 + \sum_{Y\colon S\dep Y} 2 \cdot \abs{\sA_Y}
\le 2 + 4 \cdot \sum_{i=0}^{n-2} 3^{i} = 2\cdot 3^{n-1}.
\]
\end{proof}

\begin{corollary}
For every CFG $G$ of size $n$ there is an NFA $A$ of size $2^{\BigO(n)}$ and a DFA $D$ of size $2^{2^\BigO(n)}$ with $\dc \lang(G) = \lang(A) = \lang(D)$.
\end{corollary}

\section{CFG $\to$ DFA: Double-exponential Blowup}\label{sec:debu}
As seen in the preceding section, moving from a context-free grammar $G$ representing a subword-closed language to a language-equivalent NFA $\cA$, the size of the automaton is bounded from above by $2^{O(\abs{G})}$.
For superword closures \cite{Gruber:2009:MSH:1551570.1551577} prove the same upper bound for the size of the NFA.
From both results we immediately obtain the upper bound $2^{2^{O(\abs{G})}}$ on the size of the minimal language-equivalent
DFA recognizing the sub- or superword closure of a CFG $G$. 
This bound is essentially tight as witnessed by the family of finite languages 
$$
L_k = \bigcup_{j=1}^k \{0,1\}^{j-1} \{0\} \{0,1\}^k \{0\} \{0,1\}^{k-j}.
$$
$L_k$ contains exactly all those words $w\in\{0,1\}^{2k+1}$ which contain two $0$s which are separated by exactly $k$ letters.
Using the idea of iterated squaring in order to succinctly encode the language $\{a^{2^n}\}$ as a context-free grammar (resp.\ straight-line program) of size $\BigO(n)$, also the language $L_{2^n}$ can be represented by a context-free grammar of size $\BigO(n)$. One then easily shows that the Myhill-Nerode relation w.r.t.\ $L_{2^n}$, $\dc{L_{2^n}}$, and $\uc{L_{2^n}}$, respectively, has at least $2^{2^{n}}$ equivalence classes:
        %& \vdots & & & \vdots & \\
 %& \vdots &\\
\begin{theorem}\label{thm:debu}
There exists a family of CFGs $G_n$ of size $\BigO(n)$ (generating a finite language) such that the minimal DFA accepting either
$L(G_n)$, or $\dc L(G_n)$, or $\uc L(G_n)$, has at least $2^{2^n}$ states.
\end{theorem}

\section{Equivalence of NFAs modulo Sub-/Superword Closure}
\label{sec:equiv}
As hinted at in the introduction, one application of the sub- resp.~superword closure is (in-)equivalence checking of CFGs by
regular over-approximation.
For this, we must solve the equivalence problems for NFAs representing sub/sup-word closed languages.
Naturally, the question arises how hard this is.

Let $\sA$ and $\sB$ denote NFAs over the common alphabet $\Sigma$, having $n_{\sA}$ and $n_{\sB}$ many states, respectively.
Recall that the universality problem for NFAs, i.e.\ $\lang(\sA) \stackrel{?}{=} \Sigma^\ast$,
and hence also the equivalence problem $\lang(\sA)\stackrel{?}{=}\lang(\sB)$ are \PSPACE-complete.
Only recently, it was shown in~\cite{DBLP:journals/fuin/RampersadSX12} that these problems \emph{stay} \PSPACE-complete
even when restricted to NFAs representing languages which are closed w.r.t.\ either prefixes or suffixes or factors. 
However, in~\cite{DBLP:journals/fuin/RampersadSX12} it was also shown that for subword-closed NFAs (i.e.\ $\dc \lang(\sA)= \lang(\sA)$),
universality is decidable in linear time as  $\lang(\sA) = \Sigma^*$ holds if and only if there is an SCC in $\sA$ whose labels cover all of $\Sigma$.
It is easily shown that a similar result also holds for superword-closed NFAs (i.e.\ $\uc \lang(\sA) = \lang(\sA)$):
We have $\lang(\sA) = \Sigma^*$ if and only if $\ew \in \lang(\sA)$.

In this section we show that both equivalence problems,
i.e.\ $\dc \lang(\sA) \stackrel{?}{=} \dc \lang(\sB)$ and $\uc \lang(\sA) \stackrel{?}{=} \uc \lang(\sB)$,
are \cNP-complete, hence are easier than in the general case (unless $\NP = \PSPACE$).
In the following, we write more succinctly $\sA \dceqp \sB$ and $\sA \uceqp \sB$ for these two problems.
The following lemma is easy to prove:
\begin{lemma}
\label{lem:dc-uc-NFA}
Let $\sA$ be an NFA. Define $\sA^{\dc{}}$ as the NFA we obtain from $\sA$ by adding for every transition $q\xrightarrow{a} q'$ of $\sA$ the $\ew$-transition $q\xrightarrow{\ew} q'$. 
Similarly, define $\sA^{\uc{}}$ to be the NFA we obtain by adding the loops $q\xrightarrow{a} q$ for every state $q$ and every terminal $a\in\Sigma$ to $\sA$. Then $\dc{\lang(\sA)} = \lang(\sA^{\dc{}})$ and $\uc{\lang(\sA)} = \lang(\sA^{\uc{}})$.
\end{lemma}
To prove that both $\sA \uceqp \sB$ and $\sA \dceqp \sB$ are \cNP-complete we will give a polynomial bound on the length of a {\em separating word}, i.e.\ a word $w$ in the symmetric difference of $\lang(\sA^{\dc{}})$ and $\lang(\sB^{\dc{}})$ resp.\ of $\lang(\sA^{\uc{}})$ and $\lang(\sB^{\uc{}})$.

We first show that the DFA obtained from $\sA^{\dc{}}$ resp.\ $\sA^{\uc{}}$ using the powerset construction has a particular simple structure:
\begin{lemma}\label{lem:powerset}
Let $\sA$ be an NFA. Let $\sD^{\dc{}}_{\sA}$ (resp.\ $\sD^{\uc{}}_{\sA}$) be the DFA we obtain from $\sA^{\dc{}}$ (resp.\ $\sA^{\uc{}}$) by means of the powerset construction.
For any transition $S \xrightarrow{a} T$ of $\sD^{\dc{}}_{\sA}$ ($\sD^{\uc{}}_{\sA}$) it holds that $S\supseteq T$ (resp.\ $S\subseteq T$).
\end{lemma}
Thus, the transition relation of $\sD^{\dc{}}_{\sA}$ (disregarding self-loops) can be ``embedded''
into the lattice of subsets of the states of $\sA$, which has height at most $n_{\sA}-1$.
Hence the DFA $\sD^{\dc{}}_{\sA}$ has small diameter (although even the minimal DFA for the subword closure
can be super-polynomially \emph{larger} than an NFA \cite{DBLP:journals/fuin/Okhotin10}):
\begin{corollary}
With the assumptions of the preceding lemma:
The length of the longest simple path in $\sD^{\dc{}}_{\sA}$ (resp.\ $\sD^{\uc{}}_{\sA}$) is at most $n_{\sA}-1$.
\end{corollary}
To bound the length of a shortest separating word $w$ of two NFAs w.r.t.\ sub-/superword closure,
consider the direct sum of the corresponding DFAs and observe that a run on $w$ either has to ``make progress''
in the first, or in the second DFA:
\begin{lemma}\label{lem:sep-len}
Let $\sA$ and $\sB$ be two NFAs.
If $\sA \not\equiv_{\dc{}} \sB$ (resp.\ $\sA \not\equiv_{\uc{}} \sB$), then there exists a separating word of length at most $n_{\sA}+n_{\sB}-2$.
\end{lemma}

\begin{theorem}
The decision problems $\sA \dceqp \sB$ and $\sA \uceqp \sB$ are in \cNP.
\end{theorem}
To show \cNP-hardness, recall the proof that the equivalence problem for star-free regular expressions is \cNP-hard by reduction from \TAUT:
Given a formula $\phi$ in propositional calculus, we build a regular expression $\rho$ (without Kleene stars) over $\Sigma=\{0,1\}$
that enumerates exactly the satisfying assignments of $\phi$. Hence, $\rho \in \TAUT$ iff $\lang(\rho) = \Sigma^n$ iff $\dc{\lang(\rho)} = \Sigma^{\leq n}$,
since the subword closure can only add new words of length less than $n$ (analogously for $\uc$).

\begin{theorem}
\label{thm:eq-coNP}
The decision problems $\sA \dceqp \sB$ and $\sA \uceqp \sB$ are \cNP-hard.
\end{theorem}

\section{Application to Grammar Problems}
\label{sec:application}
We apply our results to devise an approximation approach for
the well-known undecidable problem whether $\lang(G_1) = \lang(G_2)$ for two CFGs $G_1, G_2$.
Possible attacks on this problem include exhaustive search for a word in the symmetric difference $w \in (L_1 \symdiff L_2) \cap \Sigma^{\leq n}$
w.r.t.~some increasing bound $n$ e.g.~by using incremental SAT-solving \cite{DBLP:conf/icalp/AxelssonHL08}.
Unfortunately, this quickly becomes infeasible for large problems.
Previous work has successfully applied regular approximation for ambiguity detection \cite{DBLP:conf/icalp/Schmitz07,DBLP:journals/scp/BrabrandGM10}
or intersection non-emptiness of CFGs \cite{DBLP:conf/fase/LongCMM12}.

A high-level description of our approach to (in-)equivalence-checking is given in Figure \ref{fig:eq-alg}.
\begin{figure}
\begin{enumerate}
\item Compute NFAs $\sA_1$ and $\sA_2$ for the subword closures of $G_1$ and $G_2$, respectively.
\item Check, if $\lang(\sA_1) = \lang(\sA_2)$.
\begin{enumerate}
\item Case ``Not equal'': Generate a witness $w\in \lang(G_1) \symdiff \lang(G_2)$.
\item Case ``Equal'': Refine the grammars and restart at {\bf 1.}
\end{enumerate}
\end{enumerate}
\caption{Equivalence checking via subword closure approximation.}
\label{fig:eq-alg}
\end{figure}
Of course the procedure will not terminate if $\lang(G_1) = \lang(G_2)$,
so in practice a timeout will be used after which the algorithm will terminate itself and output ``Maybe equal''.
Steps (1) and (2) might take time (at most) double exponential in the size of the grammars $G_1$ and $G_2$:
Recall that the construction of Section~\ref{sec:upperbound} yields in the worst-case an NFA $\sA_{i}$ whose number of states is exponential in the size of the given CFG $G_i$.
To check if $\dc{\lang(G_1)} = \dc{\lang(G_2)}$, an on-the-fly construction of the power-set automaton for $\sA_1\times \sA_2$ can be used which terminates as soon as a set of states is reached which contains at least one accepting state of, say, $\sA_1$ but no accepting state of $\sA_2$.
Using Lemma~\ref{lem:sep-len}, we can safely terminate the exploration of simple paths if their length exceeds the bound stated in Lemma~\ref{lem:sep-len}.
In the worst case this might take time exponential in the size of $\sA_1$ and $\sA_2$, so at most double exponential in the size of $G_1$ and $G_2$.

In the following, we describe in greater detail how we generate a separating word $w'$ in $\lang(G_1)$ or $\lang(G_2)$
if we find a separating word $w\in \dc\lang(G_1)\oplus\dc\lang(G_2)$, resp.\ how we refine $G_1$ and $G_2$ if $\dc\lang(G_1)=\dc\lang(G_2)$.

\subsection{Witness Generation for $\lang(G_1) \neq \lang(G_2)$}
If our check in step (2) returns ``Not equal'' we know that $\dc{\lang(G_1)} \neq \dc{\lang(G_2)}$
and we obtain a word $w\in \dc{\lang(G_1)} \symdiff \dc{\lang(G_2)}$, w.l.o.g.~assume
in the following $w\in \dc{\lang(G_1)} \setminus \dc{\lang(G_2)}$.
This word has length linear in $\abs{\sA_1}$ and $\abs{\sA_2}$, i.e.\ at most exponential w.r.t.\ $\abs{G_1}$ and $\abs{G_2}$.

To obtain a (direct) certificate for the fact that $\lang(G_1) \neq \lang(G_2)$, we construct a superword $w'\supw w$ with $w'\in \lang(G_1)$ --
such a $w'$ is guaranteed to exist as it is the reason for $w\in\dc{\lang(G_1)}$.
Straight-forward induction on $w$ shows:
\begin{lemma}\label{lem:meins}
For $w\in\Sigma^\ast$ a DFA recognizing $\dc\lang(\{w\})$ resp.\ $\uc\lang(\{w\})$ and having at most $\abs{w}+2$ states can be constructed in time polynomial in $\abs{w}$.
\end{lemma}
We can therefore intersect $G_1$ with a DFA accepting $\uc\lang(\{w\})$, to obtain a new CFG $G_1'$ whose size is at most 
cubic in $\abs{w}$\cite{BarHillel61,DBLP:series/sci/NederhofS08}, i.e.~exponential in the size of $G_1$.
From this grammar, we can obtain in time linear in $\abs{G'_1}$ 
a shortest word $w'$ in $\lang(G'_1)= \lang(G_1)\cap \uc{\lang(\{w\})}$.
The length of $w'$ is at most exponential in $\abs{G'_1}$, i.e.\ at most double exponential in $\abs{G_1}$.

In practice, shorter witnesses are preferable, so we construct
the shortest word in $\overline{\lang(\sA_2)} \cap \lang(G_1)$.
In theory this might incur in a triple exponential blow-up resulting from complementing $\sA_2$,
but we can find a separating word $w'$ which is \emph{not} a superword of $w$ and hence is usually shorter.

\subsection{Refinement}

In case that the test in step (2) returns ``Equal'', we refine both grammars such that subsequent subword-approximations may find a
counterexample to equality.
Assume that our equivalence check yields $\dc{\lang(G_1)} = \dc{\lang(G_2)}$.
A possible refinement strategy is to cover $L:=\dc{\lang(G_1)}$ using a finite number of
regular languages $L\subseteq L' := L_0 \cup L_1 \cup \cdots \cup L_k$ and then to repeat the equivalence check
for all pairs of refined languages $\lang(G_1) \cap L_i$ and $\lang(G_2) \cap L_i$ for all $i$.
The requirement $L'\supseteq L$ protects the refinement from cutting off potential witnesses.

A simple method is covering using prefixes:
Here we generate all prefixes $p_1,\dots,p_k$ of words in $L$ of increasing length (up to some small bound
$d$ called the \emph{refinement depth}) and set $L_i:= p_i\Sigma^*$ and $L_0 = \dc{\{ p_i \mid i\in[k]\}}$.
Since $\bigcup_i L_i \supseteq L$ this strategy preserves potential witnesses and since any counterexample eventually appears as a prefix,
this yields a semi-decision procedure for grammar inequivalence.
In our experiments we disregard the finite language $L_0$ (which can also be checked by enumeration)
and only check refinement using the infinite sets $p_i\Sigma^*$ with the goal of quickly finding \emph{some} (not the shortest) distinguishing word.
This strategy is often able to tell apart different CFLs after few iterations as shown in the following.

\subsection{Implementation and Experiments}
We implemented the inequivalence check in an extension\footnote{The fork is available from \url{https://github.com/regularApproximation/newton}.} of the \fpsolve tool \cite{ELS14}.
The additional code comprises roughly 1800 lines of C++ and uses libfa\footnote{http://augeas.net/libfa/} to handle finite automata.

Our worst-case descriptional complexity results for the subword closure of CFGs 
(exponential sized NFA, double-exponential sized DFA) and our remarks on the length of possible counterexamples
might suggest that our inequivalence checking procedure is merely of academic interest.
Here we briefly show that this is not the case, and that overapproximation via subword closures
is actually quite fast in practice.

The paper \cite{DBLP:conf/icalp/AxelssonHL08} presents \cfgana,
a tool that uses SAT-solving to attack several undecidable grammar problems by exhaustive enumeration.
We demonstrate the feasibility of our approximation approach on 
several slightly altered grammars (cf.~\cite{Tratt12}) for the \textsf{PASCAL} programming language\footnote{Available from
\url{https://github.com/nvasudevan/experiment/tree/master/grammars/mutlang/acc} .}.
The altered grammars were obtained by adding, deleting, or mutating a single rule from the original grammar \cite{Tratt12}.
We used \fpsolve and \cfgana to check equivalence of the altered grammar with the original.
Both tools were given a timeout of $30$ seconds.
We want to stress that we do not strive to replace enumeration-based tools like \cfgana, but rather envision a combined approach:
Use overapproximations like the subword closure (with small refinement depth) as a quick check and 
resort to more computationally demanding techniques like SAT-solving for a thorough test.
Also note that it is not too hard to find examples where enumeration-based tools
cannot detect inequivalence anymore, e.g.~by considering grammars with large alphabet
(like C\# or Java) for which the shortest word in the language is already longer than $20$ tokens.
Here we just showcase an example where both approaches can be fruitfully combined.

Table~\ref{table:exp} demonstrates that even if our tool uses the very simple prefix-refinement (which is the main
bottleneck in terms of speed), we can successfully solve $100$ cases where \cfgana has to give up after $30$ seconds and
even in cases where both tools find a difference, \fpsolve does so much faster.
\begin{table}[th]
\begin{center}
\begin{tabular}{ccccccccc}
scenario & \# instances & \# CA & $t_{\mathrm{CA}}$ & \#FP & $t_{\mathrm{FP}}$ &  \#$(CF \wedge FP)$ & $t^{\wedge}_{\mathrm{CA}}$ & $t^{\wedge}_{\mathrm{FP}}$ \\ 
\hline
add & 700 & 190 & 17.9 & 18 & 2.43 & 8 & 10.7 & 4.97 \\ 
delete & 284 & 61 & 17.8 & 34 & 0.424 & 10 & 14.4 & 0.464 \\ 
empty & 69 & 32 & 18.7 & 1 & 1.35 & 1 & 5.62 & 1.35 \\ 
mutate & 700 & 167 & 19.1 & 100 & 1.3 & 36 & 15.8 & 2.87 \\ 
switchadj & 187 & 16 & 20.5 & 2 & 5.46 & 1 & 9.68 & 0.34 \\ 
switchany & 328 & 35 & 18 & 9 & 3.72 & 8 & 9.09 & 2.84 \\ 
\hline
\hline
$\sum$ & 2268 & 501 & -- & 164 & -- & 64 & -- & -- \\ 
\end{tabular}
\end{center}
\caption{Numbers of solved instances for different scenarios and respective average times: \#CA: solved by \cfgana, \#FP: solved by \fpsolve,
\#$(CA \wedge FP)$: solved by both tools, $t^{\wedge}_{\mathrm{tool}}$: time needed by \emph{tool} on instances from $(CA \wedge FP)$.}
\label{table:exp}
\end{table}

\section{Discussion and Future Work}
\label{sec:conclusion}
Motivated by the language-equivalence problem for context-free languages,
we have studied the problems of the space requirements of representing the subword closure of CFGs by NFAs and DFAs,
and the computational complexity of the equivalence problem of subword-closed NFAs.
We have shown how to construct from a context-free grammar $G$ an NFA accepting $\dc\lang(G)$ consisting of at most $2^{\BigO(\abs{G})}$ states -- 
a small gap between the lower bound of $\Omega(2^{\abs{G}})$ and our upper bound of $\BigO(3^{\abs{G}})$ for grammars in QNF remains for future work.
A further question is if this bound can be improved in the case of languages given by as deterministic pushdown automata.
We have further shown that the upper-bound on the size of DFA accepting $\dc\lang(G)$ of $2^{2^{\BigO(\abs{G})}}$ is tight. 
Interestingly, a binary alphabet suffices for the presented languag family $L_k$:
for instance the worst-case example of \cite{DBLP:journals/fuin/Okhotin10}, which showcases the exponential blow-up suffered when constructing an DFA for the subword closure of a language given as DFA or NFA, requires an unbounded alphabet.
We note that a unary context-free language cannot lead to this double exponential blow-up -- this follows from the proof of Theorem 3.14 in \cite{Gruber:2009:MSH:1551570.1551577} (see also Lemma~\ref{lem:meins} here).
 %assuming that $\dc\lang(G)\subseteq \{a\}^\ast$ is not trivial, there has to be word $w\in \lang(G)$ (obtain by means of an acyclic derivation) such that $\dc\lang(G)=\dc\{w\}$ with $\abs{G} \ge \abs{w}+1$;
Regarding the language-equivalence problem, we have shown that it becomes \cNP-complete when restricted to sub- resp.~superword-closed NFAs. This is somewhat surprising given the fact that it stays \PSPACE-complete for many related families (e.g.~for prefix-, suffix-, or factor-closed languages).
Finally, we have briefly described an approach to tackle the equivalence problem for CFGs using the presented results,
though much work remains to turn our current implementation into a mature tool:
In particular, since the intersection of two regular overapproximations is again a regular overapproximation,
it could be fruitful to combine the subword closure (or variants like \cite{DBLP:conf/fase/LongCMM12})
with other regular approximation techniques like \cite{MohriNederhof01}.
We also need to improve the refinement of the approximations when scaling the problem size.

\bibliographystyle{plain}
\bibliography{lit} %
\newpage
\appendix
\section{Missing proofs}
\subsection*{Proof of Lemma \ref{lem:scc}}
\begin{proof}
Since $G$ is strongly connected, $\dc{X_i} = \dc{X_j}$ for all $i,j\in[n]$, hence it suffices to show the statement for $X_1$.
Clearly, $\lang(Z) \supseteq \lang(X_1)$ hence also $\dc{Z} \supseteq \dc{X_1}$.
For the other inclusion let $w\in \dc{Z}$, i.e.~we have a word $w'$ with $w\sw w'\in \lang(Z)$ possessing some derivation
$Z \Rightarrow u_0 Z v_0 \Rightarrow u_0u_1 Z v_1 v_0 \Rightarrow \dots \Rightarrow w'$.
Since $G$ is strongly connected there must be an $X_{j_1}$ reachable from $X_1$ with $X_{j_1} \to u_0X_{k_1}v_0$ for some $Y$.
Continuing this reasoning we generate a superword of $w'$ (with some ``junk''-strings $\alpha_l,\beta_l$) by following the derivation of $w'$:
\[X_1 \Rightarrow^* \alpha_0 X_{j_1} \beta_0 \Rightarrow \alpha_0 u_0 X_{k_1} v_0 \beta_0 \Rightarrow^*
 \alpha_0 u_0 \alpha_1 u_1 X_{k_2} v_1 \beta_1 v_0 \beta_0 \Rightarrow \dots  \Rightarrow w''
\]
with $w' \sw w''$. Since, $w \sw w'$ we have $w\in \dc{X_1}$.
\end{proof}

\subsection*{Proof of Theorem \ref{thm:simpleQNF}}
\begin{proof}
The following steps achieve the desired result:
\begin{enumerate}
\item For every $x \in \Sigma \cup \{\varepsilon\}$ replace every occurrence of $x$ in a production by $A_x$ and finally add the production $A_x \to x$.
\item For every production $X \to \alpha Y \beta Z \gamma$ with $Y,Z \equiv X$
replace all productions with lhs $Y$ such that $Y\equiv X$ (i.e.~from the same SCC as $X$) by the productions
$X \to A_xX$ for all $x\in \Sigma_X$ and add $X \to A_\varepsilon$.
\item Transform the grammar into 2NF, i.e.~such that every production is of the form $X \to \alpha$ with $|\alpha| \leq 2$ (cf.~\cite{DBLP:journals/didactica/LangeL09}).
\item Contract every strongly connected component of the grammar into a univariate grammar via Lemma \ref{lem:scc}
\footnote{Here we implicitly treat nonterminals from lower SCCs as terminals, since CFLs are closed under substitution this is fine.}.
\end{enumerate}
It is easy to check that $G'$ is indeed in simple QNF, moreover steps (1) and (3) do not change the language of the grammar.
In step (2) we ensure that $\lang(X) = \Sigma_X^*$ if $X\Rightarrow^* \alpha X \beta X \gamma$ (see Lemma \ref{lem:facts-courcelle}).
Step (4) also preserves the subword closure (by Lemma \ref{lem:scc}), thus altogether $\dc{\lang(G)} = \dc{\lang(G')}$.
Step (2) reduces the size of $G$, steps (1) and (3) lead to a linear growth, and step (4) does not change the size so together
there exists a constant $c$ (independent of $G$) such that $|G'| \leq c\cdot |G|$.
\end{proof}

\subsection*{Proof of Theorem \ref{thm:upperbound}}
Before describing the proof, we state some useful definitions:
\begin{definition}
Given a nonterminal $X$ in a grammar in simple QNF with production set $P$, we define the following sets of nonterminals and terminals:
\begin{itemize}
\item $Q(X) := \{YZ\in \vars\cdot \vars : (X \to YZ \in P) \}$ (``quadratic monomials'')
\item $L(X) := \{Y\in \vars : X \to Y \in P \}$ (``linear monomials'')
\item $C_l(X) := \{Y\in \vars : X \to YX \in P \}$ (``left coefficients'')
\item $C_r(X) := \{Y\in \vars : X \to XY \in P \}$ (``right coefficients'')
\item $\Sigma_l(X) := \Sigma \cap \bigcup \{ \dc L(Y) \mid Y\in C_l(X)\}$ (``left alphabet'')
\item $\Sigma_r(X) := \Sigma \cap \bigcup \{ \dc L(Y) \mid Y\in C_r(X)\}$ (``right alphabet'')
\end{itemize}
\end{definition}
Note that $\Sigma_l(X)$ (resp.~$\Sigma_r(X)$) is simply the set of terminals reachable from any element of $C_l(X)$ (resp.~$C_r(X)$), and therefore can easily be computed.
Since $G$ is in simple QNF we have $Y \ndeppo X$ for each $Y$ with $X\dep Y$.

\begin{proof}
For every nonterminal $X$ of $G$, let $n(X) = \{ Y \mid X \deppo Y \}$ be the number of nodes reachable from $X$ in the dependency graph. 
We proceed by induction on $n(X)$.

Pick any nonterminal $X$ with $n(X)=1$. Such an nonterminal has to exist as otherwise the dependency graph would contain a nontrivial cycle.
By definition of simple QNF, $G$ can only contain a single rule rewriting $X$ which has to be of the form $X\to a$ for some $a\in \Sigma$.
Then the following NFA $\sA_X$ obviously satisfies $\dc X = \lang(\sA_X)$ and $\abs{\sA_X} \le 2\cdot 3^{n(X)-1}$:
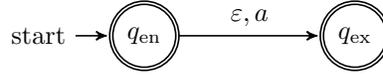
\begin{center}
\begin{tikzpicture}[->,>=stealth',shorten >=1pt,auto,node distance=2.8cm,
                    semithick]
  
\node[initial,state,accepting] (0) {$\qen$};
\node[state,accepting] (1) [right of = 0] {$\qex$};
\path (0) edge node[above] {$\ew,a$} (1);
\end{tikzpicture}
\end{center}
In the following every automaton constructed will have these special states $\qen$ and $\qex$ to which we will simply refer to as entry and exit states, respectively.

Now, let $X$ be any remaining nonterminal of $G$ with $n(X)>0$, i.e.\ there is at least one nonterminal $Y\neq X$ such that $X\dep Y$.
By virtue of Lemma~\ref{lem:facts-courcelle} and Lemma~\ref{lem:scc} we have
\[
\dc X = \Sigma_l(X)^* \left( \bigcup_{YZ \in Q(X)} \dc Y \cdot \dc Z  \cup \bigcup_{Y \in L(X)} \dc Y\right)\Sigma_r(X)^*.
\]
where by definition of simple QNF we have $Y \ndeppo X$ and $Z\ndeppo X$ implying $n(X) > n(Y), n(Z)$.
So by induction, we have already constructed for every $Y$ with $X \dep Y$ an NFA $\sA_Y$ such that $\dc Y = \lang(\sA_Y)$ i.e.\
\[
\dc X = \Sigma_l(X)^* \left( \bigcup_{YZ \in Q(X)} \lang(\sA_Y) \cdot \lang(\sA_Z)  \cup \bigcup_{Y \in L(X)} \lang(\sA_Y)\right)\Sigma_r(X)^*.
\]
It remains to construct $\sA_X$. 
To this end we use the last equality but only use at most two instances of every automaton $\sA_Y$:
Initially, we let $\sA_X$ be the disjoint union of all automata $\{ \sA^{(i)}_Y \mid i \in [2], X\dep Y\}$
where $\sA^{(1)}_Y$ and $\sA^{(2)}_Y$ denote two distinct copies of $\sA_Y$.
Here we assume that these states are suitably renamed, in particular, 
the entry and exit states of all these automata are assumed to be distinct from $\qen$ and $\qex$
so that we may add also $\qen$ and $\qex$ to the states of $\sA_X$. 
Both $\qen$ and $\qex$ are final with $\qen$ also the unique initial state of $\sA_X$.
Finally, we add additional $\ew$-transitions to $\sA_X$ to mimic the productions rewriting $X$ (see also Subsection~\ref{sec:ex-bound}):
\begin{itemize}
\item For each $YZ\in Q(X)$: 
Add $\ew$-transitions (1) from $\qen$ to the entry state of $\bl{\sA}{1}_Y$, (2) from the exit state of $\sA_Y^{(1)}$ to the entry state of $\sA_Z^{(2)}$, and (3) from the exit state of $\sA_Z^{(2)}$ to $\qex$.
\item For each $Y \in L(X)$: 
Add $\ew$-transitions (1) from $\qen$ to the entry state of $\bl{\sA}{2}_Y$, and (2) from the exit state of $\sA_Y^{(2)}$ to $\qex$.
\item For each $a\in \Sigma_l(X)$:
Add a self-loop $\qen \overset{a}{\longrightarrow} \qen$.
\item For each $a \in \Sigma_r(X)$:
Add a loop $\qex \overset{a}{\longrightarrow} \qex$.
\end{itemize}
By induction, we have $\abs{\sA_Y} \le 2\cdot 3^{n(Y)-1}$ for all $Y$ with $X\dep Y$, so $|\sA_X|$ is bounded by
\[
\abs{\sA_X} = 2 + 2 \cdot \sum_{Y\colon X\dep Y} \abs{\sA_Y} 
\le 2 + 4 \cdot \sum_{Y\colon X\dep Y} 3^{n(Y)-1}
\]
Using breadth-first search, we can assign every nonterminal $Z$ with $X\depeq^\ast Z$ a unique number $i(Y)\in [n(X)]$ such that $i(Y) \ge n(Y)$.
We then may continue:
\[
\abs{\sA_X} \le 2 + 4 \cdot \sum_{Y\colon X\dep Y} 3^{i(Y)-1} 
\le 2 + 4 \cdot \sum_{\shortstack{$\scriptstyle Z\colon X\depeq^\ast Z$\\$\scriptstyle Z\neq X$}} 3^{i(Z) - 1} 
\le 2 + 4 \cdot \sum_{i=0}^{n(X)-2} 3^{i} = 2\cdot 3^{n(X)-1}.
\]

\end{proof}

\subsection*{Proof of Theorem \ref{thm:debu}}
\begin{proof}
For $k\in\N$ consider the language $L_k$ of words $w\in\{0,1\}^{2k+1}$ such that $w_{j} = w_{j+k+1} = 0$
for some $j\in \{1,\dots,k\}$. We can write $L_k$ as
\[
L_k = \bigcup_{j=1}^k \{0,1\}^{j-1} \{0\} \{0,1\}^k \{0\} \{0,1\}^{k-j}.
\]
We in particular interested in $L_k$ for $k=2^n$.
The following CFG of size $\BigO(n)$ with $L(X'_n) = L_{2^n}$ achieves an exponential compression:
\[
\begin{array}{lcl@{\hspace{1cm}}lcl}
X'_n & \to & X_{n-1} X'_{n-1} \mid X'_{n-1} X_{n-1}\\
X'_{n-1} & \to & X_{n-2} X'_{n-2} \mid X'_{n-2} X_{n-2} & X_{n-1} & \to & X_{n-2} X_{n-2}\\
        & \vdots & & & \vdots & \\
X'_{1} & \to & X_{0} X'_{0} \mid X'_{0} X_{0} & X_{1} & \to & X_{0} X_{0}\\
X'_0 & \to & 0Y_n 0 & X_0 & \to & 0 \mid 1\\[1mm]
Y_n & \to & Y_{n-1} Y_{n-1}\\
Y_{n-1} & \to & Y_{n-2} Y_{n-2}\\
 & \vdots &\\
Y_1 & \to& Y_0 Y_0\\
Y_0 &\to & 0 \mid 1\\[1mm]
\end{array}
\]
The grammar uses repeated squaring to achieve the required compression while the ``primed'' nonterminals $X'_i$ nondeterministically
choose where to insert a word from the set $\{0\} \{0,1\}^{2^n} \{0\}$ into a word of $\{0,1\}^{2^n}$.
 %$X'_k$ represents a nonterminal with a ``marker'' that decends the binary derivation tree (of height $k$).

We show that any two words $w_1,w_2 \in \{0,1\}^{2^n}$ with $w_1\neq w_2$ are inequivalent w.r.t.~the Myhill-Nerode relation of $L_
{2^n}$ which implies that the minimal DFA for $L_{2^n}$ must have at least $2^{2^n}$ states:
Consider the first position from the right where $w_1$ and $w_2$ differ, so w.l.o.g.~we have
$w_1 = \alpha 0 \beta$ and $w_2=\alpha' 1 \beta$ for some $\alpha,\alpha',\beta \in \{0,1\}^*$.
As a distinguishing word set $v := 1^{2^{n}-\abs{\beta}}01^{2^{n}-\abs{\alpha}-1}$.
Note that 
\[w_1v = \alpha 0 \beta 1^{2^{n}-\abs{\beta}} 0 1^{2^{n}-\abs{\alpha}-1} \in L_{2^n}, \]
\[w_2v = \alpha'1 \beta 1^{2^{n}-\abs{\beta}} 0 1^{2^{n}-\abs{\alpha}-1} \notin L_{2^n}. \]
The crucial observation is that from $\abs{w_1v} = \abs{w_2v} = 2\cdot 2^n+1$ it also follows that
$w_1v \in \dc L_{2^n}$ and $w_2v \notin \dc L_{2^n}$ since the subword closure can only add new words of length at most $2\cdot 2^n$.
This shows that also the minimal DFA for $\dc L_n$ must have at least $2^{2^n}$ states.
The very same argument works for $\uc L_{2^n}$, showing that the minimal DFA for $\uc L_{2^n}$ is of size at least double-exponential in the size of the CFG for $L_{2^n}$ as well.
\end{proof}

\subsection*{Proof of Lemma \ref{lem:dc-uc-NFA}}
\begin{proof}
We start with $\dc{\lang(\sA)} = \lang(\sA^{\dc{}})$: 
Pick any $w\in \dc{\lang(\sA)}$. Then there is some $w'\supw w$ such that $w'\in \lang(\sA)$, and thus by construction also $w'\in \lang(\sA^{\dc{}})$.
That is there is an accepting run $q_0 \xrightarrow{x_0} q_1 \xrightarrow{x_1} \ldots \xrightarrow{x_l} q_{l+1}$ with $q_{l+1}\in F$ and $w'=x_0x_1\ldots x_l$ (with potentially $x_i=\ew$ for some $i$). 
Using the additional $\ew$-transitions of $\sA^{\dc{}}$ we therefore can turn this sequence into an accepting sequence for $w$ by simply replacing those $x_i$ by $\ew$ which do not occur in $w$.
For the other direction, one can reverse this argument by recalling that for any $\ew$-transition $q\xrightarrow{\ew} q'$ added to $\sA^{\dc{}}$ there is some $a\in\Sigma$ such that $q\xrightarrow{a} q'$ is a transition of $\sA$.

Consider now the second claim $\uc{\lang(\sA)} = \lang(\sA^{\uc{}})$: 
Choose some $w\in \uc{\lang(\sA)}$. Then there is some $w'\subw w$ such that $w'\in \lang(\sA) \subseteq \lang(\sA^{\uc{}})$. 
Any accepting run $q_0 \xrightarrow{x_0} q_1 \xrightarrow{x_1} \ldots \xrightarrow{x_l} q_{l+1}$ (with $q_{l+1}\in F$ and $w'=x_0x_1\ldots x_l$) of $\sA^{\uc{}}$ can then be extended to an accepting run of $\sA^{\uc{}}$ for $w$ by using the additional loops of $\sA^{\uc{}}$  to consume any letters occurring exclusively in $w$. In the other direction given an accepting run of $\sA^{\uc{}}$ we simply strip it by any loops which is guaranteed to yield an accepting run (for a scattered subword) of $\sA$ as the transition relations of $\sA$ and $\sA^{\uc{}}$ only differ in loops.
\end{proof}

\subsection*{Proof of Lemma \ref{lem:powerset}}
\begin{proof}
Recall that the state (sets) of $\sD^{\dc{}}_{\sA}$ are closed w.r.t.\ taking $\ew$-successors in $\sA^{\dc{}}$.
As $\sA^{\dc{}}$ was obtained from $\sA$ by introducing for every transition $q\xrightarrow{a} q'$ ($a\in\Sigma$) the $\ew$-transition $q\xrightarrow{\ew} q'$,
this means that, if $q\in S$, then every state reachable from $q$ in the directed graph underlying $\sA$ has to be included in $S$, too.
As for any transition $S\xrightarrow{a} T$ in $\sD^{\dc{}}_{\sA}$, $T$ is a subset of the states reachable from $S$, the claim follows.

In case of the superword closure, pick any transition $S \xrightarrow{a} T$ of $\sD^{\uc{}}_{\sA}$ and any state $q\in S$. Then by construction of $\sA^{\uc{}}$ there is the loop $q\xrightarrow{a} q$ in $\sA^{\uc{}}$ which implies that also $q\in T$ by definition of the powerset construction.
\end{proof}

\subsection*{Proof of Lemma \ref{lem:sep-len}}
\begin{proof}
Assume $\sA \not\equiv_{\dc{}} \sB$, and let $w$ be a shortest separating word. Consider the unique run of the product DFA $\sD^{\dc{}}_{\sA}\times \sD_{\sB}^{\dc{}}$ on $w=w_0w_1\ldots w_l$:
\[
(L_0,R_0) \xrightarrow{w_0} (L_1,R_1) \xrightarrow{w_1} \ldots \xrightarrow{w_l} (L_l, R_l).
\]
By the preceding lemma we then have $L_i \supseteq L_{i+1}$ and $R_i \supseteq R_{i+1}$ along the run.
As $w$ is assumed to be a shortest separating word, it has to hold that $\neg (L_i = L_{i+1} \wedge R_i = R_{i+1})$ for all $i=1,\ldots,l-1$.
In other words, we have 
\[
n_{\sA} + n_{\sB} \ge \abs{L_0} + \abs{R_0} > \abs{L_1} + \abs{R_1} > \ldots > \abs{L_l} + \abs{R_l} \ge 2
\]
from which the claim immediately follows.

In the case of the superword closure one deduces in the same way that the accepting run for a shortest separating word has to satisfy:
\[
2 \le \abs{L_0} + \abs{R_0} < \abs{L_1} + \abs{R_1} < \ldots < \abs{L_l} + \abs{R_l} \le n_{\sA} + n_{\sB}
\]
\end{proof}

\subsection*{Proof of Thorem \ref{thm:eq-coNP}}
\begin{proof}
Let $\varphi$ be a formula of propositional calculus in disjunctive normal form. We construct a regular expression which encodes all satisfying assignments of $\varphi$:

Let $x_1,x_2,\ldots,x_n$ be the propositional variables occurring in $\varphi$, and assume that $\varphi= \bigvee_{i\in[k]} C_i$ with $C_{i} = \bigwedge_{j\in[l_i} L_{i,j}$ and $L_{i,j}$ literals.
Further, we may assume that in every conjunction $C_i$ is contradiction free. We associate with every $C_i$ a simple regular expression $\rho_i$ enumerating all satisfying assignments of $D_i$:
Initially, set $\rho_i = \emptyset$. Going from $j=1$ to $j=n$, if $x_j$ occurs in $C_i$, then set $\rho_i := \rho_i 1$; if $\neg x_j$ occurs in $C_i$, set $\rho_i := \rho_i 0$; otherwise set $\rho_{i} := \rho_i (0+1)$.
Finally, set $\rho := \rho_1 + \rho_2 + \ldots + \rho_k$. Obviously, the size of $\rho$ is polynomial in the size of $\varphi$.
Further, we can compute an NFA $\sA$ from $\rho$ in time polynomial in $\abs{\rho}$, such that $\lang(\rho) =\lang(\sA)$.
Note that $\lang(\sA)=\lang(\rho)\subseteq \Sigma^n$ by construction. In particular, $\lang(\sA)=\lang(\rho) = \Sigma^n$ if and only if $\varphi$ is a tautology.

It therefore suffices to show that $\dc{\lang(\sA)} = \Sigma^{\le n}$ (resp.\ $\uc{\lang(\sA)} = \Sigma^{\ge n}$) if and only if $\lang(\sA) = \Sigma^n$.
But this is easy as the subword closure resp.\ superword closure can only add words of length less resp.\ greater than $n$.
\end{proof}

\end{document}